\documentclass{llncs}
\usepackage{amsmath,amssymb,epsfig,color,url,times}

\usepackage{algorithm, algpseudocode}
\algrenewcommand\algorithmicrequire{\textbf{\quad Input:}}
\algrenewcommand\algorithmicensure{\textbf{\quad Output:}}

\newcommand{\qedsymb}{\hfill{\rule{2mm}{2mm}}}
\renewenvironment{proof}{\begin{trivlist} \item[\hspace{\labelsep}{\bf \noindent Proof.\/}] }{\qedsymb\end{trivlist}}%
\newcommand{\bbN}{\mathbb{N}}
\newcommand{\floor}[1]{\lfloor #1 \rfloor}
\newcommand{\ceil}[1]{\lceil #1 \rceil}

\newcommand{\opt}{\mathrm{OPT}}
\newcommand{\alg}{\mathrm{ALG}}

\newcommand{\MyFrame}[1]{\noindent \framebox[\textwidth]{ \begin{minipage}{0.97\textwidth} #1 \end{minipage}}}%

\begin{document}
\title{Colored Packets with Deadlines and Metric Space Transition Cost%
\protect\footnote{Supported in part by the Israel Science Foundation (grant No. 1404/10),
by the Google Inter-university center and by The Israeli Centers of Research Excellence
(I-CORE) program, (Center  No.4/11).}}
\author{
    Yossi Azar\inst{1} %
\and
    Adi Vardi\inst{2} %
}%
\institute{
    School of Computer Science, Tel-Aviv University, Tel-Aviv 69978, Israel.,
    Israel. \\\protect\url{azar@tau.ac.il}
\and
    School of Computer Science, Tel-Aviv University, Tel-Aviv 69978, Israel.,
    Israel.\\\protect\url{adi.vardi@gmail.com}
}%
\maketitle

\vspace{-0.1cm}
\begin{abstract}
We consider scheduling of colored packets with transition costs which
form a general metric space. Let $L \geq 1$ be the minimum laxity
(the minimum difference between the expiration and the arrival time of packets).
Let MST(G) be the weight of the minimum spanning tree (MST) of a metric space G.
Let $\delta = {\rm MST(G)} / L$.
We consider the case where the laxity is large enough, i.e., $\delta < 1$.
We design a $1 - O\big(\sqrt{\delta}\big)$ competitive algorithm. Our main
result is a hardness result of $1 - \Omega\big(\sqrt{\delta}\big)$
which matches the competitive ratio of our algorithm for each
metric space separately. In particular, we improve the hardness result
of Azar at el. for a uniform metric space.
We also extend our result for a weighted directed
graph with triangle inequality. In particular, we show an algorithm
and a nearly tight hardness result.
In proving our hardness results we use an interesting non-standard embedding.
\end{abstract}

\section{Introduction}
One of the most fundamental problem in competitive analysis is the packet scheduling
problem. In this problem we are given a sequence of incoming packets.
Each packet is of unit size and has a deadline. The goal is to find a schedule
that maximizes the number of packets that were transmitted before the deadline.
The earliest deadline first (EDF) strategy is known to achieve optimal throughput
for this problem.
This paper explores a more general problem --- maximizing colored packets throughput.
In our model a switch has $m$ incoming ports and one output port. Each incoming port
is associated with an unbounded buffer. At each time unit new packets arrive to the
queues, and each packet has a deadline. The switch maintains a current incoming port through
which pending packets can be transmitted. There is a reconfiguration overhead
when the switch changes the current incoming port (i.e., time-slots dedicated to color
transition). The goal is to maximize the
number of packets that were transmitted before deadline.
Another motivation is the operation of a paint shop in a car plant.
In this model each car has to be painted before the delivery time.
If two consecutive cars have to be painted in different colors,
a colors change is required. The cost depends on the colors,
i.e., changing from pink to red costs less than changing from black to white.
We view the reconfiguration overhead as a metric space (and sometimes even as
a weighted directed graph). In particular,
this problem generalizes the problem presented in \cite{Azar_BMC}, where
the reconfiguration overhead is uniform. We consider the benefit problem
and hence $\sup_{\sigma} \!{\rm ALG}(\sigma)/{\rm OPT}(\sigma) \leq 1$.
In this paper we characterize when it is possible to achieve a $1 - o(1)$
competitive ratio.
This is done according to the specific metric space
(or the weighted directed graph)
and the laxity (the minimum difference between the expiration and the arrival time
of packets). The ideal online competitive ratio of $1 - o(1)$ is quite rare.
This is one example that this can be achieved.

\subsection{Our results} \label{subsec:OurResults}
Let $C \geq 2$ denote the number of different packet colors.
Let $L = \min_{i \in \sigma} \{d_i - r_i\} \geq 1$ denote the minimal laxity of
the packets. Also, let MST(G) be the weight of the minimum spanning tree
(MST) of the graph G. Let $\delta = {\rm MST(G)} / L$.
We consider the case where $\delta < 1$.
Our results are as follows.
\begin{itemize}
\item For a general metric space we design an algorithm with competitive
ratio of
\\ $1 - O\big(\sqrt{\delta}\big)$.

\item We show a tight hardness result of
$1 - \Omega\big(\sqrt{\delta}\big)$ for each metric space
separately. Note that for the uniform metric space our result improves
the hardness result of \cite{Azar_BMC}. Specifically, we improve their
$1 - \Omega\big(C / L\big)$ hardness result to
$1 - \Omega\big(\sqrt{C / L}\big)$
and match their algorithmic result for the uniform metric space.
\end{itemize}
We also consider the more general case where the transition costs form
a weighted directed graph with triangle inequality.
Let TSP(G) denote the weight of the minimal Traveling Salesperson
Problem (TSP) in the graph G. We note that the computation of the accurate
size is known to be NP-complete, but since all the graphs in our model
satisfy the triangle inequality we can approximate  TSP(G).
Let $\gamma = {\rm TSP(G)} / L$. We consider the case where $\gamma < 1 / \log{C}$.
Here we show nearly tight bounds for each graph. Specifically:
 \begin{itemize}
\item For any weighted directed graph G we design an algorithm with
a competitive ratio of $1 - O\big(\sqrt{\gamma}\big)$.

\item For any weighted directed graph G we prove a hardness result
of $1 - \Omega\left(\sqrt{\frac{\gamma}{\log{C}}}\right)$.
\end{itemize}
Note that when $C = 1$ EDF is optimal and hence there is no lower bound.
This is also true when ${\rm MST(G)} = 0$ (or ${\rm TSP(G)} = 0$).
Note that for any metric space G we have ${\rm MST(G)} \leq {\rm TSP(G)} < 2{\rm MST}(G)$,
hence our algorithmic results for a metric space and
a directed graph are equivalent. Note that for a directed
graph TSP(G) and directed MST(G) are not within a constant
factor. The directed MST depends on the chosen root
of the tree, but even the ratio of TSP to the maximum
over all roots of directed MST may be as large as $\Theta(n)$~\cite{Noga13}.
Fortunately TSP in a directed graph can be approximated up
to a sub-logarithmic factor ($O(\log{n}/\log{\log{n}}))$~\cite{Kaplan05,Blaser03,Frieze82,Asadpour10}

In many cases the competitive ratio of an algorithm is computed over all metric spaces. We prove  more refined results. Specifically, we show an algorithmic result and a hardness result for each metric space (or directed graph) separately.

In order to prove the hardness result we first establish it
for a star metric,  and then do it for a general metric space.
One of the techniques used in proving hardness results
for a general metric space is to show the existence of an
interesting embedding from any metric space G on nodes V
to a metric star S whose leafs correspond to V.
The embedding uses some arbitrary fixed node $v_0$ and
satisfies the following properties:
\begin{itemize}
\item The weight of S is equal to the weight
of MST(G).

\item The weight of every Steiner tree in S that contains $v_0$ is not
larger than the weight of the Steiner tree on the same nodes in G.
\end{itemize}
Note that this embedding is different from the usual embedding since we do not
refer specifically to distances between vertices. Typically, embedding
is used to prove an algorithmic result by simplifying the metric
space. By contrast, our embedding is used to prove a hardness result.

\subsection{Related work} \label{subsec:RelatedWork}
The closest model to our model
is the colored packets with deadline problem~\cite{Azar_BMC}.
In that model we are given a sequence of incoming packets. Each packet is
characterized by a color and a deadline (all packets have the same value).
The goal is to find a schedule that maximizes the number of packets that were
transmitted before the deadline, such that there is a transition
time-slot between the transmission of packets of different colors.
An algorithm with competitive ratio of $1 - O\big(\sqrt{C / L}\big)$ and hardness
result of $1 - \Omega\big(C / L\big)$ are shown in~\cite{Azar_BMC}, when
$C$ is the number of colors and $L$ is the minimum of the difference between
the expiration and the arrival time of packets.

In the bounded delay model~\cite{KesselmanLMPSS04} we are given a sequence
of incoming packets. Each packet is characterized by a value and a deadline.
The goal is to find a schedule that maximizes the number of packets that were
transmitted before the deadline. The earliest deadline first (EDF)
strategy is known to achieve an optimal throughput when all the packets
have the same value. For arbitrary values the following results are known.
A deterministic competitive algorithm of about $1 - 1 / e \approx 0.547$~\cite{EnglertW06,LiSS07}
and a hardness result of $1 / \phi \approx 0.618$
~\cite{Hajek01,ChinF03,AndelmanMZ03}. A randomized competitive algorithm of
$1 - 1 / e \approx 0.632$~\cite{BartalCCFJLST04,ChinCFJST06} and
a hardness result of $0.8$~\cite{ChinF03}.
For additional papers related to the bounded delay model
see~\cite{LiSS05,ChrobakJST07,BienkowskiCDHJJS09}.

In the FIFO queue model~\cite{KesselmanLMPSS04} we are given
a sequence of incoming packets. The packets are placed in a
FIFO queue with bounded buffer size. The challenges for an online
algorithm are to decide whether to accept or discard arriving packet,
and to choose a queue for transmission in each time unit.
The problem was studied in~\cite{AndelmanM03,BansalFKMSS04,KesselmanMS05,EnglertW06,FiatMN08}.
Another related problem is the sorting buffer problem~\cite{Racke02}.
In that problem we are given a server with unbounded capacity and an
incoming sequence of requests. Each request corresponds to a point
in a metric space. The goal is to serve all requests while minimizing
the total distance traveled by the server. This problem can be interpreted as
a multi-port device problem. This model was studied in
~\cite{Englert05,Khandekar06,BarYehuda06,Gamzu07,Englert07}.

In the proof of  the hardness result we use a non-standard embedding.
Embeddings have been studied extensively over the years~\cite{Bourgain85,Rao99}.
Much effort was invested in embedding of metric spaces into a tree metric.
Typically, a probability distribution over a family of embeddings is used,
rather then a single embedding.
Karp~\cite{karp19892k} was the first to suggest the probabilistic metric.
Bartal~\cite{Bartal96} formally defined the notion of probabilistic embedding and
proved that any probabilistic embedding of an expander graph into
a tree has a distortion of at least $\Omega(\log{n})$.
He also proved $polylog(n)$ distortion for general metric space.
Finally, Fakcharoenphol et al.~\cite{Fakcharoenphol03} showed that any
$n$-point metric space can be embedded into a distribution over dominating tree metrics with distortion O$(\log{n})$.

\subsection{Structure of the paper}
In Section \ref{sec:TheModel} we describe the model.
In Section \ref{sec:LowerBoundStar} we prove a tight hardness
result for a star metric. In Section \ref{sec:LowerBoundGraph}
we prove a tight hardness result for a general metric space.
In addition, we show the existence of a non-standard embedding
from any metric space to a star metric.
In Section \ref{sec:LowerBoundDirectedGraph} we prove a near tight
hardness result for an arbitrary weighted directed graph. In Section \ref{sec:Algorithm}
we describe the online $(1 - o(1))$-competitive algorithm for
our problem and its analysis.

\section{The Model} \label{sec:TheModel}
We formally model the problem as follows.
There is a switch of $m$ incoming ports and one output port.
Each incoming port is associated with an unbounded buffer.
When the switch changes the current incoming port from $j$ to $k$, the reconfiguration overhead
is $w(j, k)$ time-slots for $j \neq k$. Clearly, $w(j, k) = 0$ for $j = k$. We also view ports  as colors.
We are given an online sequence of packets $\sigma$. Each packet is
characterized by a triplet $(r_{i}, d_{i}, c_{i})$, where $r_{i} \in \bbN_{+}$ and $d_{i} \in \bbN_{+}$
are the respective arrival time and deadline time of the packet, and $c_{i}$ is its color.
The goal is to find a feasible schedule that maximizes the number of transmitted packets.
A feasible schedule must satisfy the following properties:
\begin{itemize}
\item In each time-slot, either we transmit a packet, or we are in color transition phase, or
this is an idle time-slot.

\item Every scheduled packet $i$ is transmitted
between time unit $r_{i}$ and time unit $d_{i}$. Otherwise, the packet is dropped.

\item Between the transmission of a packet with color $j$ and a successive packet
with color $k$ there are a $w(j, k)$ time-slots dedicated to color transition for $j \neq k$.
\end{itemize}
We view $w(j, k)$ as the weights of edges of a directed graph.
Where we change a color from $i$ to $j$, we may first change the color from $i$ to $k$ and
then from $k$ to $j$. Consequently, each such graph satisfies the triangle inequality.
\begin{itemize}
\item Typically, $w(j, k) = w(k, j)$, and hence the graph becomes undirected and can be viewed as a metric space.

\item An interesting special case is when $w(j, k) = 1$ for $j \neq k$, i.e., the uniform metric space.
This case has been considered in \cite{Azar_BMC}.

\item A special case  of general metric space is that of a star metric (which is a generalization of the uniform metric).
In a star metric, a transition from a color $i$ to a color $j$ requires $w_{i}$ + $w_{j}$ time-slots (when $w_{i}$
denotes the weight of the edge going into the node $i$). This model is equivalent to the case where the transition time to color
$j$ is $2w_{j}$.

\end{itemize}
Let $\alg(\sigma)$ $\opt(\sigma)$) denote the throughput of the online (respectively, optimal) schedule
with respect to a sequence $\sigma$.

\section{Hardness Results for a General Metric Space  and a Weighted Directed Graph}

\subsection{Hardness Result for a Star Metric} \label{sec:LowerBoundStar}

In this section we consider the case where the transition time between colors is represented
by a star metric. This is also equivalent to the case where the transition time to color $i$ is $w_{i}$.
It generalizes the model presented in \cite{Azar_BMC}, where all transition times are the same
(in particular, equal to 1). We prove that it is not possible to achieve a competitive ratio of $1 - o(1)$
when the weight of the star metric (i.e., the sum of the weights of the edges of the star) is
asymptotically larger than the minimal laxity of the packets. Applying the result to the uniform model
improves the hardness result presented in \cite{Azar_BMC} and proves that the BG algorithm from \cite{Azar_BMC}
is asymptotically optimal. The general idea is that the adversary creates packets with large deadline at
each time unit, and also blocks of packets with close deadlines. The number of packets which arrive during a block
is significantly smaller than the number of time-slots in the block.
Any online algorithm must choose between three options:
\begin{itemize}
\item Transmit mostly packets with large deadline. In this case, the online algorithm
loses packets due to EDF violation.

\item Transmit mostly packets with close deadline. In this case, the online algorithm
loses packets due to idle time.

\item Switch frequently between colors. In this case, the online algorithm
loses packets due to color transitions.
\end{itemize}

Let $w({\rm S})$ denote the weight of the star metric  S  (i.e., the sum of the weights of
the edges of  S). We define $F = \sqrt{w({\rm S})L}$.
Let $\delta = {\rm MST(G)} / L = w({\rm S}) / L$.
\begin{theorem} \label{thm:LowerBoundStar}
No deterministic or randomized online algorithm can achieve a competitive ratio better than
$1 - \Omega\big(\sqrt{\delta}\big)$ in any given star metric $\rm S$ where $\delta < 1$.
Otherwise, if $\delta \geq 1$, the bound becomes $\alpha < 1$ for some constant $\alpha$.

\end{theorem}

\begin{proof}
Let S be a given star metric with $C$ nodes.
We can assume, without loss of generality, that
$\delta < 1$, since otherwise
one may use packets with laxity of $w({\rm S})$ (i.e., $\delta = 1$), and obtain an hardness result
of $\alpha < 1$ for some constant $\alpha$.
Let type A color denote color 0 and type B color denote colors $1,...,C-1$. Let type A packet and
type B packet refer to packets with type A color and type B color, respectively.
Let $w_i$ denote the weight of the edge incident to the vertex of color $i$.
\\ We begin by describing the sequence $\sigma({\rm S}, \alg)$.
{\\\bf Sequence structure:}
There are up to $N=\frac{1}{3}\sqrt{\frac{L}{w({\rm S})}}$ blocks,
where each block consists of $3F$ time-slots.
Let $t_i = 1 + 3(i-1)F$ denote the beginning time of block $i$.
For each block $i$, where $1 \leq i \leq N$, $F$ packets of various colors
arrive at the beginning of the block.
Specifically, $\frac{w_c}{w({\rm S})-w_0}F$ type B packets $(t_i,L+t_i,c)$, for each $1 \leq c \leq C-1$
are released.
A type A packet $(t,3L,0)$ is released at each time unit $t$ in each block.
Once the adversary stops the blocks, additional packets arrive
(we call this the final event). The exact sequence is defined as follows:
\begin{enumerate}
\item $i \leftarrow 1$
\item Add block $i$ \label{subsec:AddBlockStar}
\item If with probability at least 1/4
there are at least $F/2$ untransmitted type B packets at the end of block $i$
(denoted by Condition 1),
then $L$ packets $(t_{i+1},L+t_{i+1},1)$ are released and the sequence is terminated.
See Figure \ref{figure:Case1}.
Clearly, $t_{i+1}$ is the time of the final event. Denote this by Termination Case 1.

\item Else, if with probability at least 1/4,
at most $2F$ packets were transmitted during block $i$ (denoted by Condition 2),
then $3L$ packets $(t_{i+1},3L,0)$ are released and the sequence is terminated.
Clearly, $t_{i+1}$ is the time of the final event.
See Figure \ref{figure:Case2}.
Denote this by Termination Case 2.

\item Else, if $i = N$ (there are $N$ blocks, none of them satisfied Condition 1 or 2)
$2L$ packets $(L+1,3L,0)$ are released, and the sequence is terminated.
Clearly, $L+1$ is the time of the final event.
See Figure \ref{figure:Case3}.
Denote this by Termination Case 3.

\item Else ($i < N$) then

\begin{enumerate}
\item $i \leftarrow i + 1$
\item Goto \ref{subsec:AddBlockStar}
\end{enumerate}

\end{enumerate}
We make the following
{\\\bf Observations:}
\begin{enumerate}
\item Each block consists of $3F$ time-slots. Hence, if $\alg$ transmitted at most $2F$ packets during
a block, there must have been at least $F$ idle time-slots.

\item There are up to $\frac{1}{3}\sqrt{\frac{L}{w({\rm S})}}$ blocks and each block consists of 3$\sqrt{w({\rm S})L}$
time-slots. Hence, the time of the final event is at most $L+1$.

\item Exactly one type A packet arrives at each time-slot until the final event. Hence, at most
$L$ type A packets arrive before (not including) the final event.

\item During each block, exactly $F$ type B packets arrive, which sum up to
at most $L/3$ type B packets before (not including) the final event.

\end{enumerate}
Now we can analyze the competitive ratio of $\sigma({\rm S}, \alg)$.
Consider the following possible sequences (according to the termination type):
\begin{enumerate}
\item Termination Case 1: Let $Y$ denote the number of packets in the sequence. According to the observations, the sequence
consists of at most $L$ type A packets, and at most $\frac{4}{3}L$ type B packets ($L/3$ until the
final event and $L$ at the final event). Hence, $Y \leq L + \frac{4}{3}L \leq 3L$.

\begin{itemize}
\item {\bf We bound the performance of ALG:}
At time $t_{i+1}$ there is a probability of at least 1/4 that $\alg$ has $L + F/2$ untransmitted type
B packets. Since type B packets  have laxity of $L$, $\alg$ can transmits at most $L+1$ of them, and must
drop at least $F/2 - 1$. The expected number of transmitted packets is
$$
E(\alg(\sigma)) \leq Y - \frac{1}{4}(F/2 - 1) = Y - F/8 + 1/4.
$$
\item {\bf We bound the performance of OPT$^\prime$:}
$\opt ^\prime$ transmits the packets in three stages:
\begin{itemize}
\item {\bf Type B packets that arrive before the final event:}
Recall that all type B packets in a block arrive at once in the beginning of the block.
Let $j_{k}$, $1 \leq k \leq C$, be the order of colors in the minimal TSP. In each
block $\opt ^\prime$ transmits the type B packets according to that order --- first all
packets with color $j_{1}$, then all the packets with color $j_{2}$, and so on.
It is clear that $\opt ^\prime$ needs at most $F+2w({\rm S})$ time-slots to transmits the packets
($F$ for packet transmission and $2w({\rm S})$ for color transition).
$\opt ^\prime$ transmits the packets from the beginning of the block.
Recall that $L > w({\rm S})$ and $F = \sqrt{w({\rm S})L}$. Therefore $2F > 2w({\rm S})$.
Since the block's size is $3F$, there are enough time-slots.
\item {\bf Type B packets that arrive during the final event:}
The $L$ packets $(t_{i+1},$ $L+t_{i+1},1)$ arrived during the final
release time are transmitted by $\opt ^\prime$ consecutively from time $t_{i+1}$.
$\opt ^\prime$ can transmit $L$ packets, except for one transition phase, and
hence may lose at most $2w({\rm S})$ packets. According to the observations,
the time of the final event $t_{i+1}$ is at most $L+1$.
We conclude that $\opt ^\prime$ transmits all type B packets until time unit $2L$.

\item {\bf Type A packets:}
$\opt ^\prime$ transmits the $L$ type A packets consecutively from time unit
$2L+1$. Since the deadlines are $3L$, $\opt ^\prime$ transmits all type A packets.
\end{itemize}
We conclude that $\opt(\sigma) \geq \opt ^\prime (\sigma) \geq Y - 2w({\rm S})$,
\end{itemize}
The competitive ratio is
\begin{eqnarray*}
\frac{E(\alg(\sigma))}{\opt(\sigma)} & \leq & \frac{Y - F/8 + 1/4}{Y - 2w({\rm S})} \leq \frac{3L - \frac{1}{8}\left(\sqrt{w({\rm S})L}\right) + 1/4}{3L - 2w({\rm S})} \\
& = & 1 - \Omega\left(\sqrt{\frac{w({\rm S})}{L}}\right) \ .
\end{eqnarray*}

Here the second inequality results from the fact that the number is below 1 and the
numerator and the denominator increase by the same value.

\item Termination Case 2: The sequence consists of more than $3L$ type A packets, all deadlines are
at most $3L$.

\begin{itemize}
\item {\bf We bound the performance of ALG:}
The probability that $\alg$ was idle during
$F$ time-slots is at least 1/4. Hence, the expected number of transmitted packets is
$$
E(\alg(\sigma)) \leq 3L - F/4.
$$

\item {\bf We bound the performance of $OPT ^\prime$:}
At each time unit until the final event, $\opt ^\prime$ transmits the type A packet
that arrived at that particular time unit. Consequently, from the final event
and until time unit $3L$, $\opt ^\prime$ transmits the type A packets that arrived at
the final event. Therefore, $\opt ^\prime$ transmits $3L$ type A packets, and so
$\opt(\sigma) \geq \opt ^\prime (\sigma) \geq 3L$.
\end{itemize}

The competitive ratio is
$$
\frac{E(\alg(\sigma))}{\opt(\sigma)} \leq \frac{3L - F/4}{3L} =
\frac{3L - \frac{1}{4}\left(\sqrt{w({\rm S})L}\right)}{3L} = 1 - \Omega\left(\sqrt{\frac{w({\rm S})}{L}}\right).
$$

\item Termination Case 3: the sequence consists of $3L$ type A packets, all deadlines are
at most $3L$.

\begin{itemize}
\item {\bf We bound the performance of ALG:}
Let $U_{i}$ be the event that the number of untransmitted type B
packets at the end of block $i$ is less than $F/2$. If $U_{i}$ occurs, then let
$j_{k}$, $1 \leq k \leq r$, be the type B colors transmitted by $\alg$ in block $i$.
At least $F/2$ packets that arrived in this block have to be transmitted
(recall that $F$ type B packets arrive at the beginning of each block).
Therefore,
$$
\frac{w_{j_{1}}}{w({\rm S})-w_0}F + \frac{w_{j_{2}}}{w({\rm S})-w_0}F + \cdots + \frac{w_{j_{r}}}{w({\rm S})-w_{0}}F \geq F/2,
$$
and so
$$
w_{j_{1}} + w_{j_{2}} + \cdots + w_{j_{r}} \geq \frac{w({\rm S})-w_{0}}{2}.
$$
Let $E_{i}$ be the event that more than $2F$ packets are transmitted during block $i$.
If event $U_{i-1}$ and $E_{i}$ occur, then there are at most $3F/2$ untransmitted
type B packets in the beginning of block $i$ ($F$ arrived in the beginning of the block
and at most $F/2$ from the previous block) but more than $2F$ packets were transmitted.
Therefore, at least one type A packet was transmitted during the block.
Combining the results, if $U_{i}$, $U_{i-1}$ and $E_{i}$ occur then:
\begin{itemize}
\item During block $i$ at least $(w({\rm S})-w_{0})/2$ time-slots were used for type B color
transition.

\item Type A packet was transmitted during the block.
\end{itemize}
A block $i$ is called {\it good\/} if the events $U_{i}$, $U_{i-1}$ and $E_{i}$ occur.
For any two (consecutive) good blocks the transition cost is at least
$(w(S)-w_{0})/2  + w_{0} \geq w(S)/2$.
Since none of the blocks satisfy Condition 1 or 2, it follows that
for all $i$ such that  $\frac{1}{3}\sqrt{\frac{L}{w({\rm S})}} \geq i \geq 1$ we have:
${\rm Pr}[U_{i}] \geq 3/4, {\rm Pr}[U_{i-1}] \geq 3/4,$ and ${\rm Pr}[E_{i}] \geq 3/4$.
Therefore:
\begin{align*}
& {\rm Pr}[U_{i} \cap U_{i-1} \cap E_{i}]   =   1 - {\rm Pr}[\neg (U_{i} \cap U_{i-1} \cap E_{i})]  \\
& \qquad =   1 - {\rm Pr}[\neg U_{i} \cup \neg U_{i-1} \cup \neg E_{i}] \geq 1 - 1/4 - 1/4 - 1/4 = 1/4.
\end{align*}
The sequence consists of $\frac{1}{3}\sqrt{\frac{L}{w({\rm S})}}$ blocks.
Therefore, the expected number of good blocks is  $\frac{1}{4} \cdot \frac{1}{3}\sqrt{\frac{L}{w({\rm S})}} = \frac{1}{12}\sqrt{\frac{L}{w({\rm S})}}$
and hence the expected number of disjoint pairs of blocks is
$\frac{1}{24}\sqrt{\frac{L}{w({\rm S})}}$.
Consequently, the expected number of lost packets is at least
$\frac{1}{24}\sqrt{\frac{L}{w({\rm S})}} \frac{w({\rm S})}{2}$.
We conclude that the expected number of transmitted packets is
$$
E(\alg(\sigma)) \leq 3L - \frac{w({\rm S})}{48}\sqrt{\frac{L}{w(S)}}
= 3L - \frac{1}{48}\left(\sqrt{w({\rm S})L}\right).
$$

\item {\bf We bound the performance of $OPT^\prime$:}
At each time unit until the final event, $\opt ^\prime$ transmits the type A packet
that arrived at the same time unit. Consequently, from the final event
and until time unit $3L$, $\opt ^\prime$ transmits the type A packets that arrived at
the final event. Therefore, $\opt ^\prime$ transmits 3L type A packets, and so
 $\opt \geq \opt ^\prime \geq 3L$.

\end{itemize}
The competitive ratio is
$$
\frac{E(\alg(\sigma))}{\opt(\sigma)} \leq \frac{3L - \frac{1}{48}\left(\sqrt{w({\rm S})L}\right)}{3L}
= 1 - \Omega\left(\sqrt{\frac{w({\rm S})}{L}}\right).
$$

This completes the proof of the Theorem \ref{thm:LowerBoundStar}.

\end{enumerate}

\end{proof}

\begin{corollary} \label{thm:LowerBoundUniform}
No deterministic or randomized online algorithm can achieve a competitive ratio better than
$1 - \Omega\big(\sqrt{C / L}\big)$ when all color transitions takes one unit of time
and $L > C$. Otherwise, if $L \leq C$, the bound becomes $\alpha < 1$ for some constant $\alpha$.
\end{corollary}

\begin{proof}
Let  S  be a star metric with $C$ nodes such that the weight of each edge is equal to 1/2.
Clearly, each color transition requires one time unit and $w({\rm S}) = C/2$.
Applying Theorem \ref{thm:LowerBoundStar}, we obtain the hardness result of
$1 - \Omega\big(\sqrt{C / L}\big)$.
\end{proof}
Remark: It is clear that when all color transition takes $D$ units of time the hardness result
$1 - \Omega\big(\sqrt{DC / L}\big)$.

\subsection{Hardness Result for a General Metric Space} \label{sec:LowerBoundGraph}

In this section we consider the case where the transition time is
represented by a metric space G.
A natural approach is to reduce G to a less complex graph (e.g., a star metric) and use
a  sequence similar to the sequence $\sigma({\rm S}, \alg)$ described in Section \ref{sec:LowerBoundStar}.
Note that a hardness result on a sub-graph is not a hardness result on the ambient graph.
For example, a hardness result for MST of a metric space G is not
an hardness result for G since the algorithm may use the additional
edges to reduce the transition time. Recall that $\delta = {\rm MST(G) }/ L$.
We begin by describing the requirements for embedding  G into a star metric S, such that
using $\sigma({\rm S}, \alg)$ for $\alg$ on a metric space G yields
an hardness result of $1 - \Omega\big(\sqrt{\delta}\big)$.
Then we prove that such embedding exists for any metric space G. By combining the results, we conclude that the hardness result
of $1 - \Omega\big(\sqrt{\delta}\big)$ holds for any metric space G.
\\ We begin by introducing some new definitions:

\begin{itemize}
\item We define $w({\rm T}) = \sum\limits_{e \in V} w(e)$ for a tree ${\rm T} = (V,E)$, and
let $P_{{\rm T}}(v)$ denote the parent of node $v$ in a rooted tree ${\rm T}$.

\item We define $w_{\rm S}(V) = \sum\limits_{v \in V} w(c, v)$ ($= \sum\limits_{v_i \in V} w_i$)
for a star metric ${\rm S}$ with a center $c$. It is clear that for a star ${\rm S}$ with leaves $V$, $w_{\rm S}(V) = w({\rm S})$.

\item Let ${\rm T}_{\rm G}(V)$ be the minimum weight connected component that contains the set $V$
(i.e., the minimum Steiner tree on these points) in the metric space  G.

\item Let $E({\rm G})$ be the set of edges of graph  G.
\end{itemize}
First we prove that the embedding exists for any metric space G.
\begin{theorem} \label{thm:GraphToStar}
For any given metric space ${\rm G}$ on nodes $V$ and for any vertex $v_0 \in V$
there exists a star metric ${\rm S}$ with leaves $V$ and an embedding $f: {\rm G} \rightarrow {\rm S}$
from  ${\rm G}$  to ${\rm S}$
$(f$ depends on $v_0)$ such that:
\begin{enumerate}
\item $w({\rm S}) = {\rm T}_{\rm G}(V)$  {\rm (= MST(G))}.

\item For every $V' \subseteq V$ such that $v_0 \in V'$,
$w({\rm T}_{\rm G}(V')) \geq w_{\rm S}(V').$
\end{enumerate}
\end{theorem}

\begin{proof}
See Appendix \ref{appsec:GraphToStarProof}.
\end{proof}
Now we use the embedding to prove a hardness result of $1 - \Omega\big(\sqrt{\delta}\big)$.
\begin{theorem} \label{thm:LowerBoundMetricSpace}
No deterministic or randomized online algorithm can achieve a competitive ratio better than
$1 - \Omega\big(\sqrt{\delta}\big)$ in any given metric space ${\rm G}$, where $\delta < 1$.
Otherwise, if $\delta \geq 1$, the bound becomes $\alpha < 1$ for some constant $\alpha$.
\end{theorem}

\begin{proof}
See Appendix \ref{appsec:MetricSpaceProof}.
\end{proof}
\subsection{Hardness Result for Directed Graphs} \label{sec:LowerBoundDirectedGraph}

In this section we consider the case where the transition time is
represented by a directed graphs with triangle inequality.
Let $\gamma = {\rm TSP(G)} / L$.
We prove near tight hardness result of
$1 - \Omega\left(\sqrt{\frac{\gamma}{\log{C}}}\right)$.
As in Section \ref{sec:LowerBoundStar}, the sequence consist of blocks.
Each block consist of phases. Each phase ``forces'' the online algorithm
to transmit  packets from at least 1/2 of the remaining type B colors
(colors that were not used during the previous blocks).
After $\log{C}$ phases, the online algorithm has to spend at
least TSP(G)  time-slots for color transition. This technique
guarantee that the online algorithm loses enough packets, which implies a hardness result of
$1 - \Omega\left(\sqrt{\frac{\gamma}{\log{C}}}\right)$.

\begin{theorem} \label{thm:LowerBoundGeneralGraph}
No deterministic online algorithm can achieve a competitive ratio better than
$1 - \Omega\left(\sqrt{\frac{\gamma}{\log{C}}}\right)$ in any given directed graph ${\rm G}$ with
triangle inequality, where $\gamma < 1 / \log{C}$.
Otherwise, if $\gamma \geq 1 / \log{C}$, the bound becomes
$1 - \Omega(1 / \log{C})$.
\end{theorem}

\begin{proof}
See Appendix \ref{appsec:DirectedGraphProof}.
\end{proof}

\section{Online Scheduling Algorithm for General Metric Space and Directed Weighted Graph} \label{sec:Algorithm}
In this section we design a deterministic online algorithm, for
a weighted directed or undirected graph. Without loss of
generality, we may assume that the triangle inequality holds.
Hence we can view the graph as a complete graph.
In particular, general metric spaces correspond to undirected graphs.
The algorithm achieves a competitive ratio of $1 - o(1)$ when the
minimum weight of the TSP is asymptotically small with respect to
the minimum laxity of the packets. As shown in the previous sections,
this requirement is essential in designing a $1 - o(1)$ competitive algorithm.

\subsection{The algorithm} \label{subsec:OnlineAlgo}

The algorithm is a natural extension of the BG algorithm from
\cite{Azar_BMC}. Algorithm BG works in phases of $\sqrt{CL}$ time-slots.
At each phase it collects the packets for the next phase. It transmits them
according to the colors, from color 0 to color $C-1$. Our algorithm, which we call
TSP-EDF, formally described in Figure \ref{alg:TSPEDF}, works
in phases of $K = \sqrt{{\rm TSP(G)}L}$ time-slots. In each phase the
algorithm transmits packets by colors. The order of the colors is determined
by the minimum TSP. The algorithm achieves a competitive ratio of
$1 - 3\sqrt{{\rm TSP(G)}/ L}$. Clearly, finding the  TSP(G)  is not known to require a
polynomial time (it is an NP hard problem). To make our algorithm polynomial we
use an approximation of the TSP (e.g., 2-approximation for undirected graphs
and ($\log{C}/\log{\log{C}}$)-approximation for directed ones). This will replace
 TSP(G) by  MST(G)  for undirected graphs and by ${\rm  TSP(G)}\log{C}/\log{\log{C}}$
for directed graphs.

\begin{figure}[htbp]
\MyFrame{
\smallskip In each phase $\ell = 1, 2, \ldots$,\  do
\begin{itemize}
\item Reduced the deadline of each untransmitted packet $(r,d,c)$
from $d$ to $K \floor{d/K}$.

\item Let $S^\ell$ be the collection of untransmitted packets
such that their reduced deadline was not exceed. Let $S^{\ell,K}$
be the K-length prefix of EDF schedule (according to the modified
deadline) of $S^\ell$. Let $S_j^{\ell,K} \subseteq S^{\ell,K}$ denote the
subset of packets having color $c_j$.
\\ Let $i_{1}$, $i_{2}$, ..., $i_{C}$ denote the order of the colors in the
minimal TSP (or approximation).

\item $\rho_\ell$ initially consists of all packets of $S_{i_{1}}^{\ell,K}$
scheduled consecutively, then all packets of $S_{i_{2}}^{\ell,K}$ scheduled
consecutively, and so on.

\item $\rho_\ell$ is modified so that $w(v_{i}, v_{j})$ color
transition time-slots are added between each two successive color
groups $S_i^{\ell,K}, S_{j}^{\ell,K}$

\item $\rho_\ell$ is modified so that its length will be exactly $K$.
If the length of $\rho_\ell$ is more than $K$, then the last packets
are dropped. If its length is less than $K$, then it is affixed
with idle time-slots.

\item The packets are transmitted according to $\rho_\ell$.

\end{itemize}

}\caption{Algorithm TSP-EDF.} \label{alg:TSPEDF}
\end{figure}

\smallskip \noindent {\bf Analysis.}
The analysis is similar to the analysis in \cite{Azar_BMC}.
The full proof is in Appendix \ref{appsec:OnlineAlgoAnalysis}.

\bibliographystyle{abbrv}
\bibliography{SchedulingColoredPackets}

\newpage

\appendix
\section{Proofs} \label{appsec:Proofs}

\subsection{Proof of Theorem \ref{thm:GraphToStar}} \label{appsec:GraphToStarProof}
We prove the theorem by describing a star metric with the required properties.
Let  G be a given metric space on nodes $V$ and a leaf $v_0 \in V$.
Let  T  be the MST for  G  created by means of Prims algorithm
with the root $v_0$. Let  S  be a star metric with leaves V such that for each
$u \in V$, $w_u = w(u, P_{{\rm T}}(u))$. Clearly, $w_{v_0} = 0$.
We prove that  S  and $v_0$ satisfy the theorem's properties:
{\\\bf Property 1:} Clearly, $w({\rm S}) = w({\rm T})$, and since T is a MST for G,
$w({\rm S}) = w({\rm T}) = {\rm MST(G)} = {\rm T}_{\rm G}(V)$.
{\\\bf Property 2:} We have to prove that for every $V' \subseteq V$ such that $v_0 \in V'$,
$w({\rm T}_{\rm G}(V')) \geq w_S(V').$
Let $V' = \{v_0, v_{i_1}, ..., v_{i_r-1}\}$.
Recall that we defined $w_u = w(u, P_{\rm T}(u))$. Clearly,  $w_{\rm S}(V') = \sum_{j=1}^{r-1} w(v_{i_j}, P_{{\rm T}}(v_{i_j}))$.
Hence it  suffices to prove that $w({\rm T}_{\rm G}(V')) \geq \sum_{j=1}^{r-1} w(v_{i_j}, P_{\rm T}(v_{i_j}))$.
The proof is based on the following idea. We begin with the minimum Steiner tree that contains
$V'$ (${\rm T}_{\rm G}(V')$). Then we transform it to an MST on all vertices
by running Prim from $v_0$ and replacing the Steiner tree's edges with Prim's edges.
We prove that each time the algorithm adds an edge $e$ that corresponds to an edge in
$w_S(V')$ it deletes an edge $e'$ from ${\rm T}_{\rm G}(V')$ such that $w(e) \leq w(e')$. Note that we also add edges incident to vertices not in $V'$ in order to maintain a tree.
The weights of these edges are not counted.
Since the algorithm starts with ${\rm T}_{\rm G}(V')$ and finishes
with T, this proves that the property holds (recall that the weight of the edges of  S
is determined by the weight of the edges of  T).
The exact description of our algorithm, called the Embed-Prim algorithm, is provided in Figure \ref{alg:EmbedPrim}.

\begin{figure}[htbp]
\MyFrame{

\begin{enumerate}
\item ${\rm T}' \leftarrow {\rm T}_{\rm G}(V')$ \label{alg:Init}

\item $i \leftarrow 1$ \label{alg:AfterInit}

\item $V_{\rm new} = {v_0} ($let $u_0 = v_0)$

\item Repeat until $V_{\rm new} = V$

\begin{enumerate}

\item Choose an edge $e_i = (w,u_i)$ with minimal weight such that $w$ is in
$V_{\rm new}$ and $u_i$ is not (ties are broken by id).

\item Add $u_i$ to $V_{\rm new}$

\item If $u_i \notin V'$ then add $e_i$ and $u_i$ to ${\rm T}'$. \label{alg:Case1}
\begin{enumerate}
\item Else, if $e_i \in E({\rm T}')$ then replace the edge $e_i$ with the same edge $e_i$
(needed only for the proof). \label{alg:Case2}

\item Else, if $e_i \notin E({\rm T}')$ then \label{alg:Case3}
\begin{enumerate}
\item Add $e_i$ to $T'$.

\item Let $C'$ be the cycle created by adding $e_i$ to ${\rm T}'$. Let $e'$ be the edge with
maximal weight on $C'$ such that $e' \notin \{e_1, ..., e_{i}\}$ and
$e' \cap \{u_0, ..., u_{i-1}\} \neq \emptyset$ (i.e., $e'$ is the maximal among the edges in $C'$
that was added before step \ref{alg:AfterInit}, and one of their nodes is in $V_{\rm new}$.
In Lemma \ref{lemma:EmbedPrimExistEdge} we prove that such an edge always exists).
Remove $e'$ from ${\rm T}'$.
\end{enumerate}
\end{enumerate}

\item $i \leftarrow i + 1$

\end{enumerate}
\end{enumerate}

}\caption{Algorithm Embed-Prim.} \label{alg:EmbedPrim}
\end{figure}

First we show the correctness of Embed-Prim:
\begin{lemma} \label{lemma:EmbedPrimExistEdge}
Let $C'$ be the cycle created in step \ref{alg:Case3}. There exists at least
one edge $e'$ that belongs to $C'$, such that $e' \notin \{e_1, ..., e_{i}\}$ and $e' \cap \{u_0, ..., u_{i-1}\} \neq \emptyset$.
\end{lemma}
\begin{proof}
Note that a cycle is created in step \ref{alg:Case3} since adding an edge to a tree always creates a cycle.
Similar to Prim, the edges that Embed-Prim adds after step \ref{alg:Init} do not create a cycle.
Therefore $C'$ must contains edges added in step \ref{alg:Init}.
At least one of these edges must touch one of the vertices
$\{u_0, ..., u_{i-1}\}$
\end{proof}

\begin{lemma} \label{lemma:EmbedPrimTree}
After each step of Embed-Prim, ${\rm T}'$ is a tree which contains $V'$.
\end{lemma}

\begin{proof}
At the beginning ${\rm T}'$ is ${\rm T}_{\rm G}(V')$, which is a tree that contains $V'$.
We never remove vertices and hence ${\rm T}'$ always contains $V'$.
Whenever we add an edge that creates a cycle we open the cycle
by removing an edge from it.
\end{proof}

Now we claim that Embed-Prim satisfies the following invariant:
\begin{lemma} \label{lemma:EmbedPrimInvariant}
Each time Embed-Prim adds an edge $e$ that corresponds to an edge in
$w_{\rm S}(V')$, it deletes an edge $e'$ from ${\rm T}_{\rm G}(V')$ such that $w(e) \leq w(e')$.
\end{lemma}
\begin{proof}
Step \ref{alg:Case1} is irrelevant, since the edge does not correspond to an edge
in $w_{\rm S}(V')$ (the vertex that was added by Embed-Prim is not in $V'$).
In step \ref{alg:Case2}, $w(e) = w(e')$.
In step \ref{alg:Case3}, since Embed-Prim could have added edge $e$', but did choose the edge $e$ instead,
$w(e) \leq w(e')$ (recall that Embed-Prim always chooses the edge with the minimal weight).
\end{proof}

\par Now we are ready to prove that S satisfies the second property of the embedding.
By the definition of Prim $e_i = (u_i, P_{{\rm T}}(u_i))$.
Hence, $\sum_{j=1}^{r-1} w(e_{i_j}) \!=\! \sum_{j=1}^{r-1} w(v_{i_j}, P_{{\rm T}}(v_{i_j}))$.
Let $e'_i$ be the edge deleted from $T'$ when edge $e_i$ was added (steps \ref{alg:Case2}, \ref{alg:Case3}). Then
$$
w_{\rm S}(V') = \sum_{j=1}^{r-1} w(v_{i_j}, P_{{\rm T}}(v_{i_j}))
=\sum_{j=1}^{r-1} w(e_{i_j})
\leq \sum_{j=1}^{r-1} w(e'_{i_j})
\leq w({\rm T}_{\rm G}(V')).
$$
where the first equality follows from the definition, the first inequality result from the invariance, and the last inequality follows from the definition.

\subsection{Proof of Theorem \ref{thm:LowerBoundMetricSpace}} \label{appsec:MetricSpaceProof}
Let  G  be a given metric space on nodes $V$. We use the embedding from Theorem \ref{thm:GraphToStar}.
Let ${\rm S}$, $v_0$ be the output of the embedding. Let $\sigma({\rm S}, \alg)$ be the sequence described
in Theorem \ref{thm:LowerBoundStar}, when $v_0$ is type A color and the other colors are
type B. Recall that, by definition, $F = \sqrt{w({\rm S})L}$.
We use $\sigma$ for $\alg$ on G.
We can assume, without loss of generality, that $\delta < 1$ since otherwise
one may use packets with laxity of  MST(G)  (i.e., $\delta = 1$), and obtain a hardness result
of $\alpha  < 1$ for some constant $\alpha$.
Consider the following possible cases, similar to the proof of Theorem \ref{thm:LowerBoundStar}.

\begin{enumerate}
\item In Termination Case 1 there exists a block $i$ such that, with probability at least 1/4, at the
end of the block there are at least $F/2$ untransmitted type B packets. In Theorem
\ref{thm:LowerBoundStar} we proved that:
\begin{itemize}
\item The sequence consists of up to $3L$ packets.

\item The expected number of packets $\alg$ missed is at least $F/8 - 1/4$.

\item $\opt$ missed up to ${\rm TSP(G)} \leq 2{\rm MST(G)}$ packets.
\end{itemize}
Therefore, the competitive ratio depends only on $F$,  MST(G)  and $L$:
\begin{eqnarray*}
\frac{E(\alg(\sigma))}{\opt(\sigma)} &\leq& \frac{3L - F/8 + 1/4}{3L - 2{\rm MST(G)}} = \frac{3L - \frac{1}{8}\left(\sqrt{w({\rm S})L}\right) + 1/4}{3L - 2{\rm MST(G)}} \\
&=& \frac{3L - \frac{1}{8}\left(\sqrt{{\rm MST(G)}L}\right) + 1/4}{3L - 2{\rm MST(G)}} = 1 - \Omega\left(\sqrt{{\rm MST(G) }/ L}\right) \ .
\end{eqnarray*}
Here the second equality results from the fact that $w({\rm T}) = {\rm MST(G)}$.

\item In Termination Case 2 there exists a block $i$ such that, with probability at least 1/4, at most $2F$ packets
were transmitted during the block. In Theorem \ref{thm:LowerBoundStar} we proved that:
\begin{itemize}
\item At most $3L$ packets can be transmitted.

\item The expected number of packets $\alg$ missed is at least $F/4$.

\item $\opt ^\prime$ transmitted 3L type A packets.
\end{itemize}
Therefore, $\opt \geq \opt ^\prime = 3L$ and the competitive ratio depends only on $F$ and $L$:
\begin{eqnarray*}
\frac{E(\alg(\sigma))}{\opt(\sigma)} &\leq& \frac{3L - F/4}{3L} = \frac{3L - \frac{1}{4}\left(\sqrt{w({\rm S})L}\right)}{3L} \\
&=& 1 - \Omega\left(\sqrt{w({\rm S}) / L}\right) = 1 - \Omega\left(\sqrt{{\rm MST(G)} / L}\right) \ .
\end{eqnarray*}
Here the last equality results from the fact that $w({\rm S}) = {\rm MST(G)}$.

\item In Termination Case 3 $\alg$ transmitted type A packet and at least $F/2$ type B packets at
each block. In Theorem \ref{thm:LowerBoundStar} we proved that:
\begin{itemize}
\item At most $3L$ packets can be transmitted.

\item The expected number of packets $\alg$ missed in each block
due to color transitions is at least $\frac{1}{8} \frac{w({\rm S})}{2}$.

\item $\opt ^\prime$ transmitted $3L$ type A packets. Therefore $\opt \geq \opt ^\prime = 3L$.
\end{itemize}
By the first property required by this theorem, each sequence of color transitions in
G requires more transition time than in S. Therefore, the expected number of packets
$\alg$ missed per block is at least $\frac{1}{8} \frac{w({\rm S})}{2}$.
Since the number of blocks is $\frac{1}{3}\sqrt{\frac{L}{w({\rm S})}}$, we conclude that the competitive ratio is:
\begin{align*}
\frac{E(\alg(\sigma))}{\opt(\sigma)} &\leq  \frac{3L - \left(\frac{1}{3}\sqrt{L/w(S)} \right)\frac{w({\rm S})}{16}}{3L}\\
&= 1 - \Omega\left(\sqrt{w({\rm S}) / L}\right) = 1 - \Omega\left(\sqrt{{\rm MST(G)} / L}\right) \ .
\end{align*}
Here the last equality result from the fact that $w(S) = MST(G)$.
\end{enumerate}

\subsection{Proof of Theorem \ref{thm:LowerBoundGeneralGraph}} \label{appsec:DirectedGraphProof}

We begin by describing the packet sequence $\sigma(G, \alg)$.
{\\\bf Sequence structure:}
We assume, without loss of generality, that $\gamma < 1 / \log{C}$, since otherwise
one may use packets with laxity of ${\rm TSP(G)}/\log{C}$ (i.e., $\gamma = 1 / \log{C}$), and obtain a hardness result
of $1 - \Omega\big(\frac{1}{\log{C}}\big)$.
We define $H = \sqrt{ {{\rm TSP(G)}L}/{\log{C}}}$,
$N = \frac{1}{5}\sqrt{{L}/({\rm TSP(G)} \log{C})}$. There are up to $N$ blocks.
Type A packets $(t,3L,0)$ are released at each time unit $t$ in each block.
Let $t_i$ denote the beginning time of block $i$.
The first block starts at time-slot 1 and blocks follow one after another.
A block consist of phases: start phase, regular phases and possibly an end phase.
The start phase consist of $2H \log{C}$ time-slots.
Every block that does not satisfy Condition $(I_1)$ (defined later in the paper)
has an end phase with $2H \log{C}$ time-slots.
The number of regular phases is at least 1 and at most $\log{C}$, and depends on the behavior of
 $\alg$. Each regular phase consists of $H$ time-slots.
We denote by $t_{i,j}$ the beginning time of regular phase $j$ in block $i$.
The first regular phase starts after the start phase,
and regular phases follow one after another.
At the beginning of the first regular phase, $H$ packets arrive of various
colors. Specifically, ${H}/( C-1)$ type B packets $(t_{i,1}, L + t_{i,1}, c)$
for each $1 \leq c \leq C-1$ are released.
Once the adversary stops the blocks (we call these the final event), additional
packets arrive. The exact sequence is defined as follows:

\begin{enumerate}
\item $i \leftarrow 1$, $j \leftarrow 1$
\item Add regular phase $j$ of block $i$:
Let $c_{k}$, $1 \leq k \leq r$, be the colors that were not transmitted during the previous
regular phases. $H$ packets of various colors arrive at once.
\\ Specifically, $H / r$ packets ($t_{i,j}$, $L + t_{i,j}$, $c_{k}$) for $1 \leq k \leq r$
are released.
 \label{subsec:AddBlockGeneralGraph}
\item If there are at least $H/2$ untransmitted type B packets at the end of the phase
(denoted as Condition $(I_1)$), $L$ packets
$(t_{i,j+1}, L + t_{i,j+1}, 1)$ are released and the sequence is terminated.
Clearly, $t_{i,j+1}$ is the time of the final event.
Call this Termination Case 1.

\item Else, during the regular phases packets from all type B colors were transmitted.
We complete the block by an end phase, and consider the following cases:

\begin{enumerate}
\item If no type A packet was transmitted during the start phase
or during the end phase (denoted as Condition $(I_2)$), then 3L packets
$(t_{i+1,1},3L,0)$ are released, and the sequence is terminated.
Clearly $t_{i+1,1}$ is the time of the final event.
Call this Termination Case 2.

\item Else, if $i = N$ (there are $N$ blocks and none of them satisfied Condition $(I_1)$ or $(I_2)$).
Let $t$ be the time-slot consecutive to the end of the last block.
$2L$ packets $(t,3L,0)$ are released, and the sequence is terminated.
Clearly, $t$ is the time of the final event.
Call this Termination Case 3.

\item Else ($i < N$, start new block)

 \begin{enumerate}
 \item $i \leftarrow i + 1$, $j \leftarrow 1$

 \item Goto \ref{subsec:AddBlockGeneralGraph}.
 \end{enumerate}

\end{enumerate}

\item Else (start new phase)
\begin{enumerate}
 \item $j \leftarrow j + 1$

 \item Goto \ref{subsec:AddBlockGeneralGraph}.
 \end{enumerate}

\end{enumerate}
\begin{lemma}
There are at most $\log{C}$ regular phases in each block.
\end{lemma}

\begin{proof}
During each regular phase at least half of the remaining type B colors
(i.e., type B colors that were not transmitted during previous regular phases)
are transmitted. It is clear that after at most $\log{C}$ regular phases, all
type B packets are transmitted.
\end{proof}
We make the following
{\\\bf Observations:}
\begin{enumerate}
\item No type B packets arrive during the start phase or in the end phase of a block.

\item Since there are at most $\frac{1}{5}\sqrt{ {L}{\rm TSP(G) }\log{C}}$ blocks,
and each block consists of at most $5H \log{C}$ time-slots
($2H \log{C}$ during the start phase, $2H \log{C}$ during
the end phase, up to $H \log{C}$ during the phases stage),
the time of the final event is at most $L+1$.

\item Exactly one type A packet arrives at each time-slot until the final event,
which sums to at most $L$ type A packets before (not including) the final release time.

\item In each block, the number of type B packets released is at most 1/5 of its size
(type B packets are released only during the regular phases).
Hence, there are at most $L/5$ type B packets before (not including) the final event.
\end{enumerate}
Now we can analyze the competitive ratio of $\sigma({\rm G}, \alg)$.
Consider the following possible sequences (according to the termination type):
\begin{enumerate}
\item Termination Case 1: Let $Y$ denote the number of packets in the sequence.
According to the observations, the sequence
consists of at most $L$ type A packets, and at most $6L/5$ type B packets ($L/5$ until the
final event and $L$ at the final event). Hence, $Y \leq L + (6L/5) \leq 3L$.

\begin{enumerate}
\item {\bf We bound the performance of ALG:}
At time $t_{i,j+1}$ $\alg$ has $L + H/2$ untransmitted type B packets. Since type B packets
have laxity of $L$, $\alg$ can transmits at most $L+1$ of them and drop at least $H/2 - 1$.
The number of transmitted packets is
$$
\alg(\sigma) \leq Y - H/2 + 1.
$$

\item {\bf We bound the performance of $OPT ^\prime$:}
$\opt ^\prime$ transmits the packets in three stages:
\begin{itemize}
\item {\bf Type B packets that arrive before the final event:}
Recall that no type B packets arrive during the end phase.
Let  ${\rm G}_{i,j}$ be the spanning subgraph of G that contains only the
colors arrived in regular phase $j$ of block $i$.
Let $c_{i,j,k}$, $1 \leq k \leq r_{i,j}$, be the order of colors in the minimal TSP
for   ${\rm G}_{i,j}$. In each regular phase (e.g. regular phase $i$ in block $j$) $\opt ^\prime$ transmits
the type B packets according to that order --- first all the packets with color
$c_{i,j,1}$, then all the packets with color $c_{i,j,2}$, and so on.
It is clear that $\opt ^\prime$ needs at most $H+{\rm TSP(G)}$ time-slots to transmits the packets ($H$ for
packet transmission and TSP(G)  for color transition).
$\opt ^\prime$ transmits the packets ordered by arrival time. Each regular phase consist of $H$
time-slots, but $\opt ^\prime$ needs up to $H+{\rm TSP(G)}$ time-slots to transmit the packets. Therefore, at the end
of the last regular phase in each block there are at most ${\rm TSP(G)} \log{C}$ untransmitted
packets. These packets are transmitted in the end phase.
There are enough time-slots, since $H \log{C} > {\rm TSP(G)} \log{C}$
(recall that $L > {\rm TSP(G)} \log{C}$ and $H = \sqrt{ {{\rm TSP(G)}L}/{\log{C}}}$).
In the last block $\opt'$ missed up to ${\rm TSP(G)} \log{C}$ packets (since there is no end phase).

\item {\bf Type B packets that arrive during the final event:}
The $L$ packets $(t_{i,j+1},$ $L+t_{i,j+1},1)$ that arrived during
the final event are transmitted by $\opt ^\prime$ consequtively from time $t_{i,j+1}$. $\opt ^\prime$
can transmits $L$ packets except for one transition period, and hence may lose at most TSP(G)
packets.

\item {\bf Type A packets:}
According to the observations, the time of the final event $t_{i,j+1}$ is at most $L+1$.
We conclude that $\opt ^\prime$ transmits type B packets until time unit $2L$, and type A packets
between $2L+1$ and $3L$.
\end{itemize}
Hence:
$$
\opt(\sigma) \geq \opt ^\prime (\sigma) \geq Y - (\log{C}+1){\rm TSP(G)}
$$
\end{enumerate}
The competitive ratio is

\begin{eqnarray*}
\frac{\alg(\sigma)}{\opt(\sigma)} & \leq & \frac{Y - H/2 + 1}{Y - (\log{C}+1){\rm TSP(G)}} \leq
\frac{3L - \frac{1}{2}\left(\sqrt{\frac{{\rm TSP(G)}L}{\log{C}}}\right) + 1}{3L - (\log{C}+1){\rm TSP(G)}} \\
& = & 1 - \Omega\left(\sqrt{\frac{{\rm TSP(G)}}{L \log{C}}}\right) \ .
\end{eqnarray*}

Here the second inequality results from the fact that the number is below 1 and the
numerator and the denominator increase by the same value.

\item Termination Case 2: The sequence consists of more than $3L$ type A packets, all deadlines are
at most $3L$.
\begin{itemize}

\item {\bf We bound the performance of ALG:}
$\alg$ did not transmit type A packets during the start phase, or during
the end phase. Condition $(I_1)$ guarantees that at the end of each regular phase there are
up to $H$/2 untransmitted type B packets. Recall that no type B packets arrive during the
start phase or in the end phase of a block. If $\alg$ did not transmit a type A packet
during the $2H \log{C}$ time-slots of the start phase, there are more than
$H \log{C}$ idle time-slots during that period (it had only $H/2$
untransmitted type B packets). By a symmetric argument, if $\alg$ did not transmit
a type A packet during the $2H \log{C}$ time-slots of the end phase,
there are more than $H \log{C}$ idle time-slots during that period.
We conclude that the number of transmitted packets is
$$
\alg(\sigma) \leq 3L - H \log{C}.
$$
\item {\bf We bound the performance of $OPT ^\prime$:}
At each time unit until the final event, $\opt ^\prime$ transmits the type A packet
that arrived at that particular time unit. From the final event
and until time unit $3L$, $\opt ^\prime$ transmits the type A packets that arrived at
the final event. Therefore, $\opt ^\prime$ transmits $3L$ type A packets and $\opt \geq \opt ^\prime = 3L$.
\end{itemize}
The competitive ratio is
\begin{eqnarray*}
\frac{\alg(\sigma)}{\opt(\sigma)} & \leq & \frac{3L - H \log{C}}{3L} =
\frac{3L - \sqrt{\frac{{\rm TSP(G)}L}{\log{C}}} \log{C}}{3L} \\
& = & 1 - \Omega\left(\sqrt{\frac{{\rm TSP(G)} \log{C}}{L}}\right) \ .
\end{eqnarray*}

\item Termination Case $3$: The sequence consists of $3L$ type A packets, all deadlines are
at most $3L$.
\begin{itemize}
\item {\bf We bound the performance of ALG:}
Every block satisfies Conditions $(I_1)$ and $(I_2)$. Therefore, type A packet is
transmitted during the start phase of each block. Then packets from
all the type B colors are transmitted during the regular phases and then a type A packet
is transmitted during the end phase. It is clear that during each block $\alg$ spent
at least  TSP(G)  time-slots for color transitions. Since there are
$\frac{1}{5}\sqrt{ {L}/({{\rm TSP(G) }\log{C}})}$ blocks, at least
$\frac{1}{5}\sqrt{ {{\rm TSP(G)}L}/{\log{C}}}$ time-slots were
used for color transition. It follows that
 $\alg(\sigma) \leq 3L - \frac{1}{5}\sqrt{ {\rm TSP(G)L}/{\log{C}}}$.
\item {\bf We bound the performance of $OPT ^\prime$:}
At each time unit until the final event, $\opt ^\prime$ transmits the type A packet
that arrived at that particular time unit. From the final event
and until time unit $3L$, $\opt ^\prime$ transmits the type A packets that arrived at
the final event. Therefore, $\opt ^\prime$ transmits $3L$ type A packets, and $\opt \geq \opt ^\prime = 3L$.
\end{itemize}
The competitive ratio is
$$
\frac{\alg(\sigma)}{\opt(\sigma)} \leq \frac{3L - \frac{1}{5}\sqrt{\frac{{\rm TSP(G)}L}{\log{C}}}}{3L} =
1 - \Omega\left(\sqrt{\frac{{\rm TSP(G)}}{L \log{C}}}\right).
$$
\end{enumerate}
This complete the proof of Theorem \ref{thm:LowerBoundGeneralGraph}.

\subsection{Analysis of the algorithm TSP-EDF} \label{appsec:OnlineAlgoAnalysis}
First we need to demonstrate that the output schedule $\rho$ is feasible.
Specifically, we need  to prove that every scheduled packet $i$
is transmitted during the time frame $[r_{i}, d_{i}]$, and that
there is a color transition of length $w(i,j)$ between the transmission of
any two successive packets with different colors $i$ and $j$.

\begin{lemma} \label{lemma:Correctness}
The algorithm ${\rm TSP-EDF}$ generates a valid schedule.
\end{lemma}
\begin{proof}
The results follows from description of the algorithm. The exact proof is
similar to the proof of Lemma 3.1 in \cite{Azar_BMC}.
\end{proof}

Now we analyze the performance guarantee of the algorithm.
We first define two input sequences $\sigma'$ and $\widetilde{\sigma}$,
which are modifications of $\sigma$. The input sequence $\sigma'$
consists of all packets in $\sigma$, but modifies the color of
packets to a fixed color $c'$. Specifically,
each packet $(r,d,c) \in \sigma$ defines a packet $(r, d, c') \in
\sigma'$. The input sequence $\widetilde{\sigma}$ consists of all
packets in $\sigma$ such that a packet $(r,d,c) \in \sigma$ gives
rise to a packet $(K \ceil{r/K}, K \floor{d/K}, c')
\in \widetilde{\sigma}$, where $c'$ is a fixed color.
Hence, all packets in $\widetilde{\sigma}$ have the
same color, and the release and deadline times of
each packet in $\widetilde{\sigma}$ are aligned with the start/end time of
the corresponding phase so that the span of each packet is
fully contained in the span of that packet according to $\sigma$.
Here that the \emph{span} of a packet $(r,d,c)$ is defined as the time frame $[r,d]$.

\begin{lemma} \label{lemma:SigmaTildaEq}
$\opt(\widetilde{\sigma}) = \alg(\widetilde{\sigma})$.
\end{lemma}
\begin{proof}
Note that algorithm TSP-EDF has three modification with respect to EDF:
\begin{itemize}
\item Packet  deadline times are modified to $K \floor{d/K}$.

\item Packet  release times are modified to $K \ceil{r/K}$
(because in each phase only packets released during previous phases
are transmitted).

\item Color transition time-slots are added between the transmission
of packets from different colors.

\end{itemize}

The release and deadline times of the packets in $\widetilde{\sigma}$
are aligned and all the packets have the same color. Hence, $\alg$'s
schedule is identical to EDF's schedule. Since EDF is optimal
for sequences that consist of packets with one color,
$\opt(\widetilde{\sigma}) = \alg(\widetilde{\sigma})$.
\end{proof}

\begin{lemma} \label{lemma:OPTPertubation}
$\opt(\widetilde{\sigma}) \geq \left(1 - 2\sqrt{{\rm TSP(G) }/ L}\right)\opt(\sigma')$.
\end{lemma}
\begin{proof}
The notion of \emph{$\lambda$-perturbation}, defined in \cite{Azar_BMC}, is as follows:
An input sequence $\widehat{\delta}$ is a
\emph{$\lambda$-perturbation} of $\delta$ if $\widehat{\delta}$
consists of all packets of $\delta$, and each packet
$(\widehat{r},\widehat{d}) \in \widehat{\delta}$ corresponding to packet
$(r,d) \in \delta$ satisfies $\widehat{r} - r \leq \lambda$ and $d -
\widehat{d} \leq \lambda$.
\\ By definition, $\widetilde{\sigma}$ is $K$-perturbation of $\sigma'$,
and the colors of all packets are identical. Hence,   Theorem
2.2 from \cite{Azar_BMC} yields the following inequality
\end{proof}

\begin{lemma} \label{lemma:TRANSPenalty}
$\alg(\sigma) \geq \big(1 - \sqrt{{\rm TSP(G)} / L}\,\big)\alg(\sigma')$.
\end{lemma}
\begin{proof}
The difference between the schedule TSP-EDF generated for
$\sigma$ and the schedule it generates for $\sigma'$ is that
packets might be dropped at the end of each phase in $\sigma$
due to color transition. The worst case for $\sigma$ is when
there are no idle time-slots in any of the phases of $\sigma'$.
Otherwise, the idle time-slots might be used for color transition.
Therefore, there are at most
$\lceil \alg(\sigma') / K \rceil - 1$ phases in which algorithm
TSP-EDF transmits packets (the $-1$ term is due to the fact
that the algorithm does not transmit any packet during the first
phase). Since there are no more than  TSP(G)  color transitions
in each phase, we obtain the following inequality:
\begin{eqnarray*}
\alg(\sigma) & \geq & \alg(\sigma') - \left(\lceil \alg(\sigma') / K \rceil - 1\right){\rm TSP(G) }\\
& = & \alg(\sigma') - \left(\left \lceil \frac{\alg(\sigma')}{\sqrt{{\rm TSP(G)}L}} \right \rceil - 1\right){\rm TSP(G)}\\
& \geq & \left(1 - \sqrt{{\rm TSP(G) }/ L}\right)\alg(\sigma') \ .
\end{eqnarray*}
\end{proof}

We are now ready to prove the main theorem of this section.

\begin{theorem} \label{thm:Approximation}
The algorithm {\rm TSP-EDF} attains a competitive ratio of $1 - 3\sqrt{{\rm TSP(G)} / L}$.
\end{theorem}
\begin{proof}
Using the previously stated results, we obtain that
\begin{eqnarray*}
\alg(\sigma) & \geq & \left( 1 - \sqrt{{\rm TSP(G) }/ L} \right) \alg(\sigma')
= \left( 1 - \sqrt{{\rm TSP(G) } / L} \right) \alg(\widetilde{\sigma}) \\
& = & \left( 1 - \sqrt{{\rm TSP(G) } / L} \right) \opt(\widetilde{\sigma}) \\
& \geq & \left( 1 - \sqrt{{\rm TSP(G) } / L} \right) \left( 1 - 2\sqrt{{\rm TSP(G) } / L} \right) \opt(\sigma') \\
& \geq & \left( 1 - 3\sqrt{{\rm TSP(G) } / L} \right) \opt(\sigma) \ .
\end{eqnarray*}
The first inequality results from Lemma~\ref{lemma:TRANSPenalty}. The
first equality follows by the definition of the algorithm. The second
equality holds by lemma~\ref{lemma:SigmaTildaEq}. The second
inequality results from Lemma~\ref{lemma:OPTPertubation}.
Finally, the last inequality holds because $\sigma'$ is similar $\sigma$,
but all packets have the same color. This implies that any schedule
feasible for $\sigma$ is also feasible for $\sigma'$, and thus
$\opt(\sigma') \geq \opt(\sigma)$.~
\end{proof}

\section{Figures} \label{appsec:Figures}
\begin{figure}[htbp]
\begin{center}
\scalebox{0.40}
{
\includegraphics{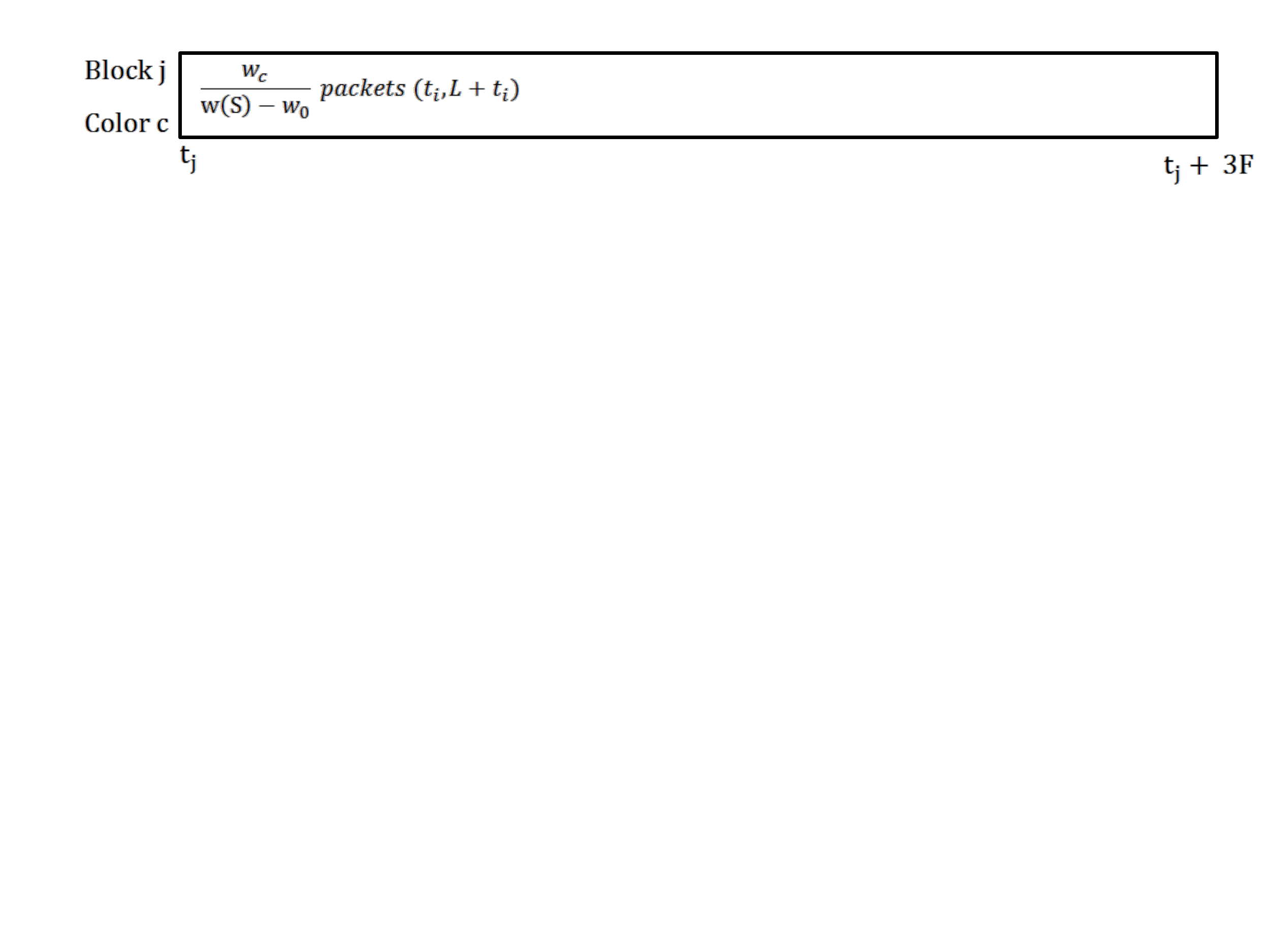}
}
\end{center}
\vspace{-190pt}
\caption{Block's structure.
The pair $(r,d)$ represent release time $r$ and deadline $d$.
Note that all the packets arrive at once in the beginning of the block.}
\label{figure:Block}
\end{figure}

\begin{figure}[htbp]
\begin{center}
\scalebox{0.40}
{
\includegraphics{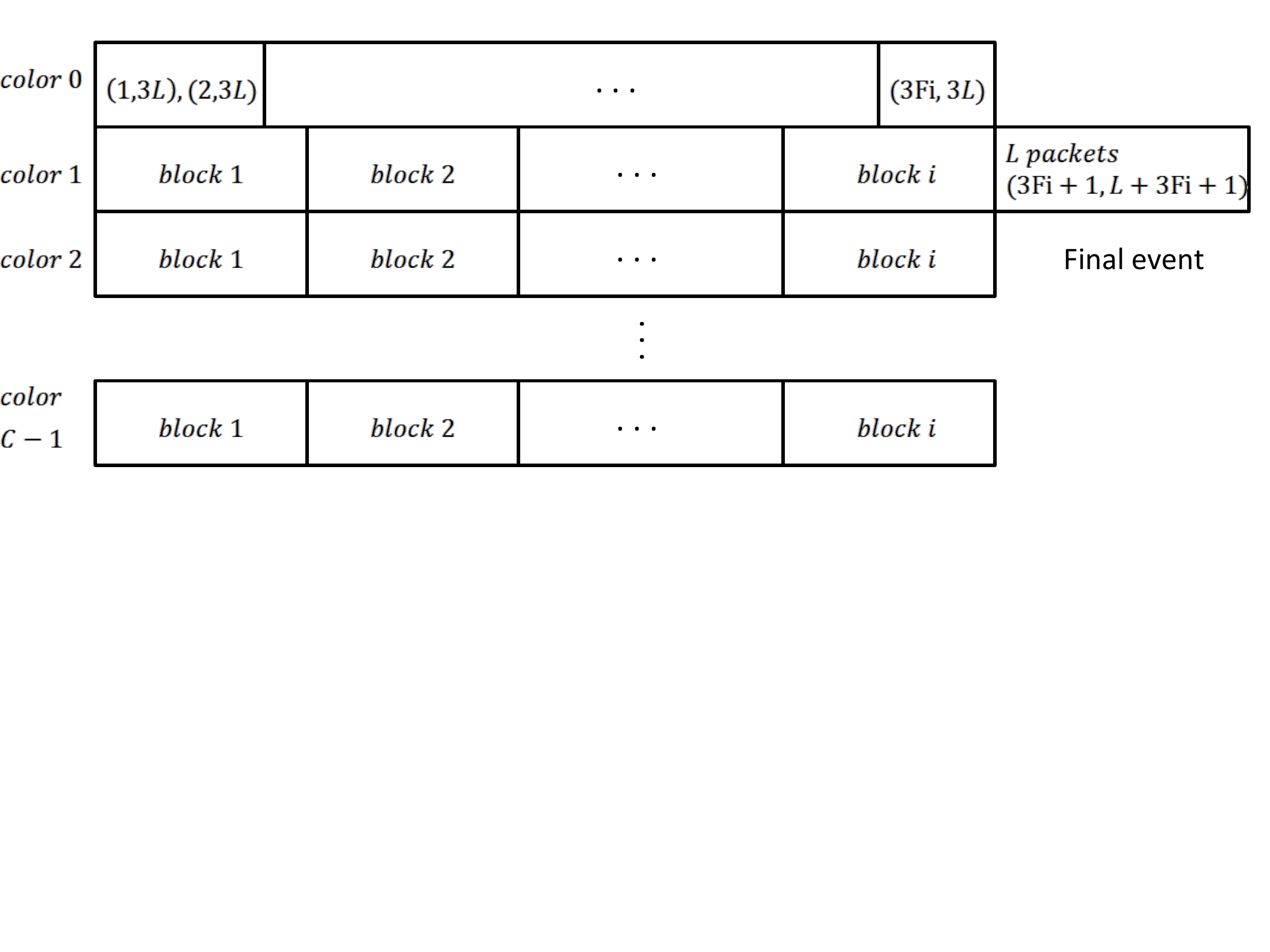}
}
\end{center}
\vspace{-110pt}
\caption{Sequence structure for Termination Case 1. See Figure \ref{figure:Block}
for blocks structure.}
\label{figure:Case1}
\end{figure}

\begin{figure}[htbp]
\begin{center}
\scalebox{0.40}
{
\includegraphics{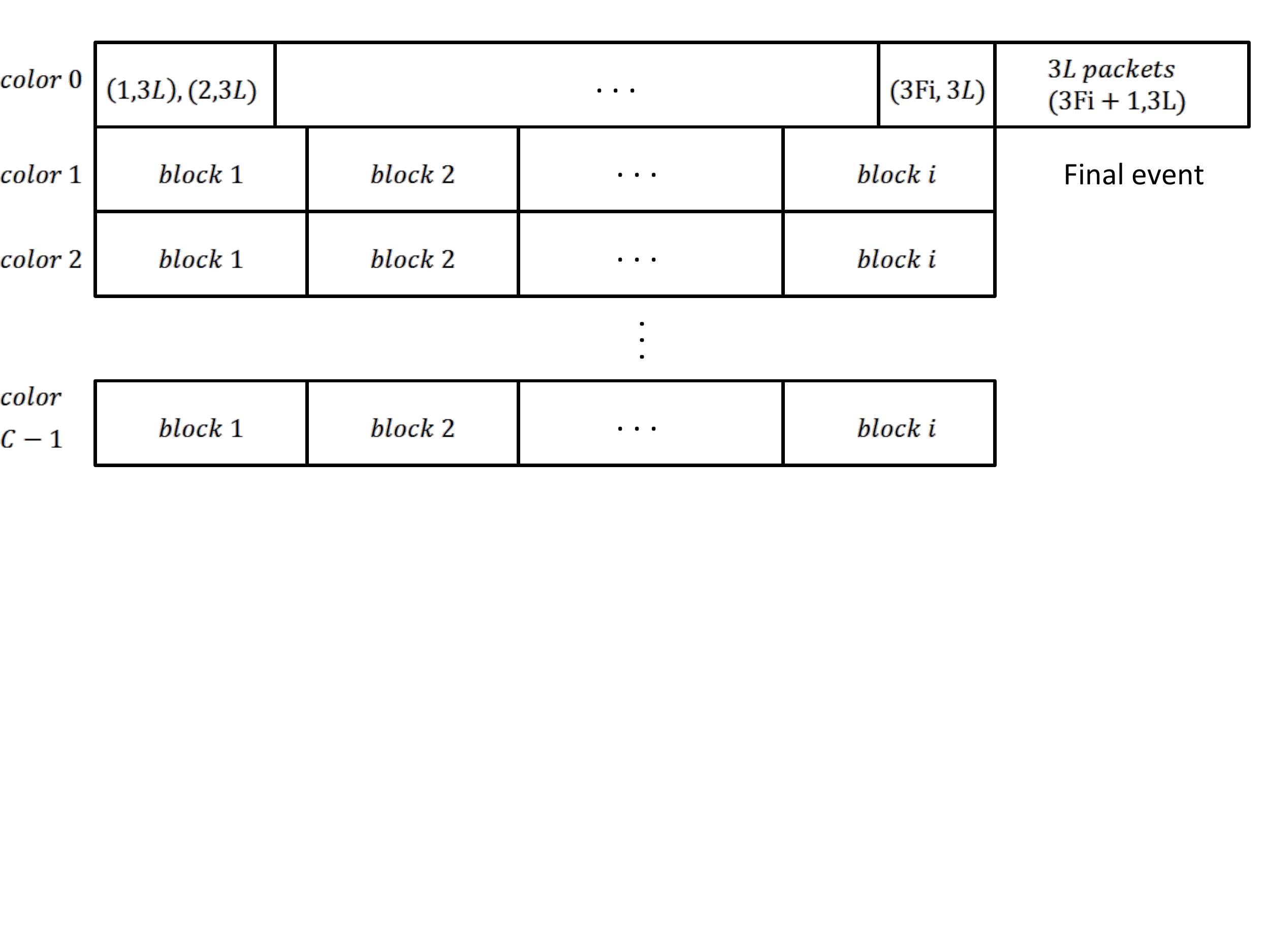}
}
\end{center}
 \vspace{-120pt}
\caption{Sequence structure for Termination Case 2. See Figure \ref{figure:Block}
for blocks structure.}
\label{figure:Case2}
\end{figure}

\begin{figure}[htbp]
\begin{center}
\scalebox{0.40}
{
\includegraphics{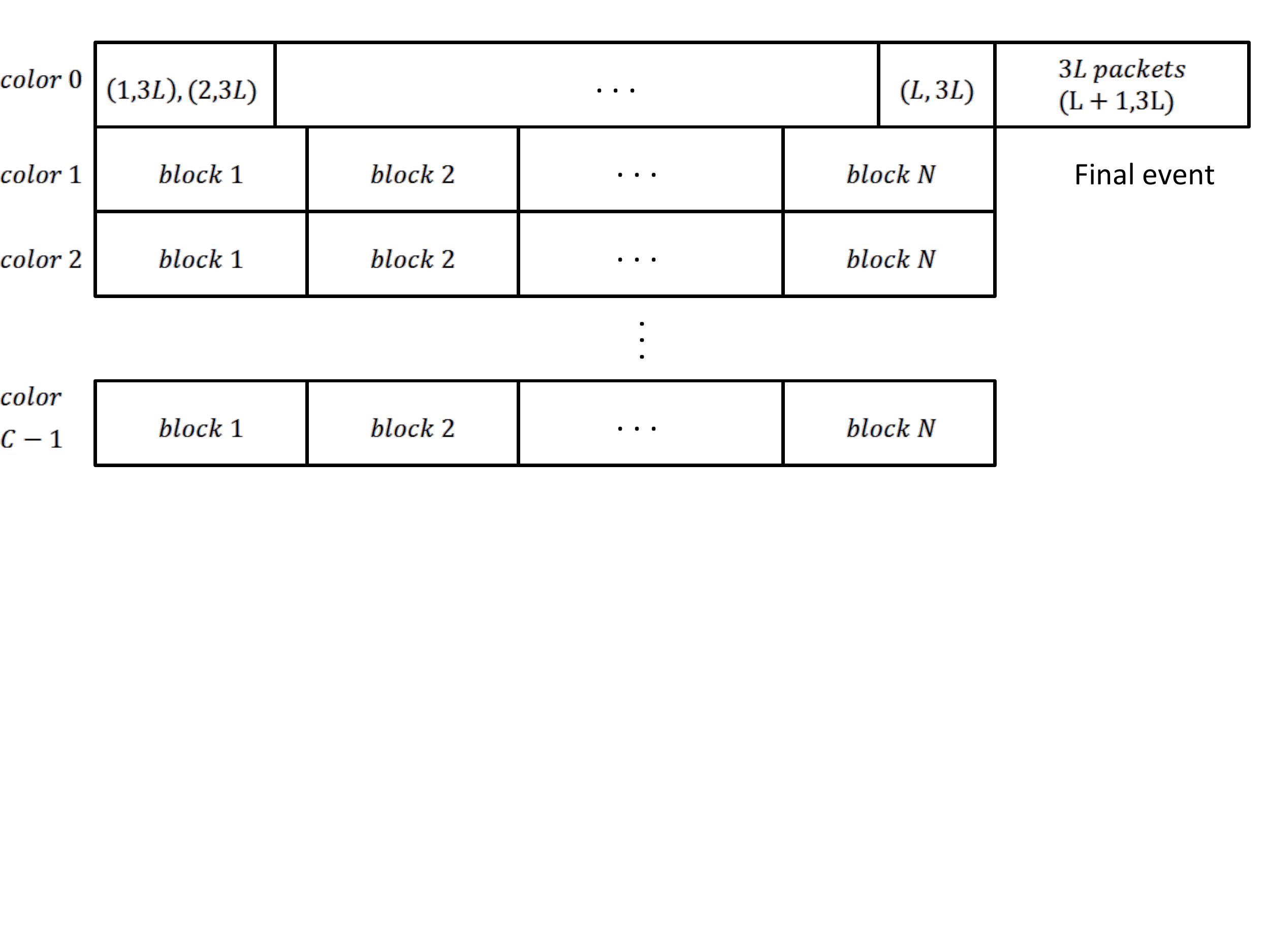}
}
\end{center}
\vspace{-120pt}
\caption{Sequence structure for Termination Case 3. Recall that $N = \frac{1}{3}\sqrt{\frac{L}{w({\rm S})}}$.
See Figure \ref{figure:Block} for blocks structure.}
\label{figure:Case3}
\end{figure}

\end{document}